\documentclass[conference]{IEEEtran}
\IEEEoverridecommandlockouts
\usepackage{cite}
\usepackage{amsmath,amsfonts,amssymb,amsbsy,amstext,amsthm}
\usepackage[ruled]{algorithm}
\usepackage{algorithmic}
\usepackage{graphicx}
\usepackage{textcomp}
\usepackage{bm}
\usepackage{xcolor}
\usepackage{subfig}
\usepackage[nolist,printonlyused]{acronym}

\begin{acronym}
\acro{TV}{television}
\acro{CPU}{central processing unit}
\acro{CPRI}{common public radio interface}
\acro{PLF}{predecessor leader following}
\acro{WHO}{world health organization}
\acro{PL}{platoon leader}
\acro{LTE-V2V}{long term evolution V2V}
\acro{PM}{platoon member}
\acro{CAM}{cooperative awareness message}
\acro{D2D}{device-to-device}
\acro{TVWS}{\acs*{TV} white spaces}
\acro{RSRQ}{reference signal received quality}
\acro{RSSI}{received signal strength indicator}
\acro{SWOT}{strengths, weaknesses, opportunities, and threats}
\acro{AP}{access point}
\acro{MS}{mobile station}
\acro{ID}{identity}
\acro{NCL}{neighbor cell list}
\acro{WLAN}{wireless local area network}
\acro{WiMAX}{worldwide interoperability for microwave access}
\acro{LTE}{long term evolution}
\acro{DMRS}{demodulation reference signal}
\acro{UL}{uplink}
\acro{DL}{downlink}
\acro{DPC}{dirty paper coding}
\acro{TTI}{transmission time interval}
\acro{CSI}{channel state information}
\acro{CoMP}{coordinated multi-point}
\acro{ISM}{industrial, scientific and medical}
\acro{3GPP}{3rd generation partnership project}
\acro{M2M}{machine-to-machine}
\acro{MTC}{machine type communication}
\acro{eNB}{evolved node B}
\acro{PRB}{physical resource block}
\acro{PF}{proportional fair}
\acro{RRA}{radio resource allocation}
\acro{MCS}{modulation and coding scheme}
\acro{SCM}{spatial channel model}
\acro{SNR}{signal-to-noise ratio}
\acro{PDF}{probability density function}
\acro{GPS}{global positioning system}
\acro{P2P}{peer to peer}
\acro{MAC}{medium access control}
\acro{OFDM}{orthogonal frequency division multiplexing}
\acro{CCL}{common cell list}
\acro{5G}{fifth generation}
\acro{FDD}{frequency division duplex}
\acro{TDD}{time division duplex}
\acro{EPA}{equal power allocation}
\acro{MR}{maximum ratio}
\acro{PhD}{doctor of philosophy}
\acro{3G}{third generation}
\acro{4G}{fourth generation}
\acro{ANACOM}{autoridade nacional de comunicações}
\acro{AF}{array factor}
\acro{ARPU}{average return per user}
\acro{AHP}{analytic hierarchy process}
\acro{B3G}{beyond \acs*{3G}}
\acro{BSc}{bachelor of science}
\acro{BV}{beamforming vector}
\acro{BWA}{broadband wireless access}
\acro{CAPEX}{capital expenditure}
\acro{COGEU}{cognitive radio systems for efficient sharing of TV white spaces in European context}
\acro{CR}{cognitive radio}
\acro{CV}{combining vector}
\acro{DEnKF}{deterministic ensemble Kalman filter}
\acro{DFT}{discrete Fourier transform}
\acro{DSRC}{dedicated short-range communication}
\acro{DVB-H}{handled digital video broadcast}
\acro{DVB-T}{terrestrial digital video broadcast}
\acro{DTT}{digital terrestrial \acs*{TV}}
\acro{ETSI}{European telecommunications standards institute}
\acro{FCC}{federal communications commission}
\acro{FP7}{seventh framework programme}
\acro{HM}{high mobility}
\acro{HSPA}{high speed packet access}
\acro{IA}{initial access}
\acro{ICT}{information and communication technologies}
\acro{IEEE}{institute of electrical and electronics engineers}
\acro{DySPAN-SC}{\acs*{IEEE} DySPAN standards committee}
\acro{ITU}{international telecommunications union}
\acro{LAN}{local area network}
\acro{LS}{least squares}
\acro{LSFD}{large-scale fading decoding}
\acro{MBMS}{multimedia broadcast multicast services}
\acro{MVNO}{mobile virtual network operator}
\acro{MT}{mobile terminal}
\acro{MF}{matched filter}
\acro{MSc}{master of science}
\acro{NPV}{net present value}
\acro{NRA}{national regulator authority}
\acro{OFCOM}{office of communications}
\acro{OMP}{orthogonal matching pursuit}
\acro{OPEX}{operational expenditure}
\acro{PHY}{physical}
\acro{PMSE}{programme making and special events}
\acro{PT}{Portugal}
\acro{RRS}{reconfigurable radio systems}
\acro{QoE}{quality of experience}
\acro{SAE}{system architecture evolution}
\acro{SDR}{software defined radio}
\acro{SLA}{service level agreement}
\acro{SVD}{singular-value decomposition}
\acro{TETRA}{trans European trunked radio access}
\acro{ULA}{uniform linear array}
\acro{UHF}{ultra high frequency}
\acro{UK}{United Kingdom}
\acro{UMTS}{universal mobile telecommunications system}
\acro{UPA}{uniform planar array}
\acro{UTRA}{\acs*{UMTS} terrestrial radio access}
\acro{USP}{unique selling point}
\acro{USA}{United States of America}
\acro{WG}{working group}
\acro{Wi-Fi}{wireless fidelity}
\acro{WRC}{world radiocommunication conference}
\acro{WSD}{white space device}
\acro{OM}[O\&M]{operational \& maintenance}
\acro{IRR}{internal rate of return}
\acro{EKF}{extended Kalman filter}
\acro{EnKF}{ensemble Kalman filter}
\acro{EPC}{evolved packet core}
\acro{E-UTRAN}{evolved universal terrestrial radio access network}
\acro{MME}{mobility management entity}
\acro{PCRF}{policy and charging rules function}
\acro{S-GW}{service gateway}
\acro{P-GW}{packet gateway}
\acro{EIRP}{equivalent isotropically radiated power}
\acro{CBR}{constant bit rate}
\acro{RRM}{radio resource management}
\acro{HetNet}{heterogeneous network}
\acro{CEPT}{European conference of postal and telecommunications administrations}
\acro{ANATEL}{agência nacional de telecomunicações}
\acro{RAT}{radio access technology}
\acro{RAN}{radio access network}
\acro{SC}{single connection}
\acro{FS}{fast switching}
\acro{DC}{dual connectivity}
\acro{C-RAN}{cloud \acl{RAN}}
\acro{NR}{new radio}
\acro{mmW}{millimeter wave}
\acro{KF}{Kalman filter}
\acro{KPI}{key performance indicator}
\acro{BB}{branch and bound}
\acro{BER}{bit error rate}
\acro{BS}{base station}
\acro{CDF}{cumulative distribution function}
\acro{CFN}{channel Frobenius norm}
\acro{DCP}{difference of convex functions program}
\acro{GTEL}{wireless telecommunications research group}
\acro{GUB}{global upper bound}
\acro{GLB}{global lower bound}
\acro{HMSINR}{highest minimum \ac{SINR}}
\acro{HPBW}{half power beamwidth}
\acro{HWCR}{highest \ac{WCR}}
\acro{HWSR}{highest weighted sum rate}
\acro{IBC}{interference broadcast channel}
\acro{IC}{interference channel}
\acro{KKT}{Karush-Kuhn-Tucker}
\acro{LB}{lower bound}
\acro{LiDAR}{light detection and ranging}
\acro{LOS}{line of sight}
\acro{LP}{linear programming}
\acro{MAP}{a maximum a posteriori probability}
\acro{MILP}{mixed integer linear programming}
\acro{mTRP}{multi transmission and reception point}
\acro{MIMO}{multiple input multiple output}
\acro{MINLP}{mixed integer nonlinear programming}
\acro{MIP}{mixed integer programming}
\acro{MISO}{multiple input single output}
\acro{ML}{machine learning}
\acro{MMSE}{minimum mean squared error}
\acro{MRC}{maximum ratio combining}
\acro{MRT}{maximum ratio transmission}
\acro{MSE}{mean squared error}
\acro{MU-MIMO}{multi-user multiple input multiple output}
\acro{MU}{multi-user}
\acro{NLP}{nonlinear programming}
\acro{NLOS}{non-line of sight}
\acro{NMSE}{normalized mean squared error}
\acro{NP}{non-polynomial time}
\acro{QoS}{quality of service}
\acro{SCA}{successive convex approximation}
\acro{SINR}{signal-to-interference-plus-noise ratio}
\acro{SISO}{single input single output}
\acro{SIMO}{single input multiple output}
\acro{SV}{steering vector}
\acro{TU}{typical urban}
\acro{UAV}{unmanned aerial vehicles}
\acro{UE}{user equipment}
\acro{UKF}{unscented Kalman filter}
\acro{JSFRA}{joint space-frequency resource allocation}
\acro{UB}{upper bound}
\acro{WINNER}{wireless world initiative new radio}
\acro{WMMSE}{weighted minimum mean squared error}
\acro{WCR}{weighted common rate}
\acro{V2V}{vehicle-to-vehicle}
\acro{V2I}{vehicle-to-infrastructure}
\acro{V2P}{vehicle-to-pedestrian}
\acro{V2X}{vehicle-to-everything}
\acro{RF}{radio frequency}
\acro{RL}{reinforcement learning}
\acro{RRV}{rate-relaxation variable}
\acro{DQN}{deep Q-network}
\acro{ADMM}{alternating direction method of multipliers}
\acro{SDP}{semidefinite programming}
\acro{PARAFAC}{parallel factor analysis}
\acro{V2N}{vehicle-to-network}
\acro{BALS}{bilinear alternating least-squares}
\acro{TALS}{trilinear alternating least-squares}
\acro{RMSE}{root mean square error}
\acro{ALS}{alternating least-squares}
\acro{5G-PPP}{5G infrastructure public private partnership}
\acro{SLNR-MAX}{maximum signal-to-leakage-plus-noise ratio}
\acro{ReLU}{rectifier linear unit}
\acro{ZF}{zero-forcing}
\acro{RUE}{resource usage efficiency}
\acro{5G-StoRM}{5G stochastic radio channel for dual mobility}
\acro{UMi}{urban micro}
\acro{DB}{database}
\acro{KNN}{K-nearest neighbors}
\acro{PCA}{principal component analysis}
\acro{CIR}{channel impulse response}
\acro{EE}{energy efficiency}
\acro{SE}{spectral efficiency}
\acro{LTE-A}{long term evolution advanced}
\acro{NOMA}{non-orthogonal multiple access}
\acro{FD}{full-duplex}
\acro{OFDMA}{orthogonal frequency-division multiple access}
\acro{PDPR}{pilot-to-data power ratio}
\acro{AWGN}{additive white Gaussian noise}
\acro{CFO}{carrier frequency offset}
\acro{SER}{symbol error rate}
\acro{M-QAM}{multilevel quadrature amplitude modulation}
\acro{B5G}{beyond fifth generation}
\acro{SSB}{synchronization signal block}
\acro{gNB}{gNodeB}
\acro{RSRP}{reference signal received power}
\acro{PSS}{primary synchronization signal}
\acro{SSS}{secondary synchronization signal}
\acro{PBCH}{physical broadcast channel}
\acro{PCI}{physical cell identity}
\acro{RACH}{random access channel}
\acro{RAR}{random access response}
\acro{BPSK}{binary phase shift keying}
\acro{TRP}{transmission reception point}
\acro{LP-MMSE}{local partial MMSE}
\acro{mMIMO}{massive multiple input multiple output}
\acro{MIMO}{multiple input multiple output}
\acro{GP}{geometric programming}
\acro{i.i.d.}{independent and identically distributed}
\acro{SIC}{successive interference cancellation}
\acro{LHS}{left hand side}
\acro{SOC}{second-order cone}
\acro{SOCP}{second-order cone program}
\acro{FD}{full-duplex}
\acro{SWIPT}{simultaneous wireless information and power transfer}
\acro{WPT}{wireless power transfer}
\acro{HD}{half-duplex}
\acro{SI}{self-interference}
\acro{IAI}{inter-AP interference}
\acro{IUI}{inter-user interference}
\acro{LMMSE}{linear minimum mean squared error}
\acro{InH}{indoor hotspot}
\acro{FPC}{fractional power control}
\end{acronym} 
\def\BibTeX{{\rm B\kern-.05em{\sc i\kern-.025em b}\kern-.08em
    T\kern-.1667em\lower.7ex\hbox{E}\kern-.125emX}}

\DeclareMathAlphabet{\mathppl}{T1}{ppl}{m}{it}
\DeclareMathAlphabet{\mathphv}{T1}{phv}{m}{it}
\DeclareMathAlphabet{\mathpzc}{T1}{pzc}{m}{it}

\DeclareMathOperator{\tr}{tr}

\newcommand{\Herm}[1]{{#1}^{\mathrm{H}}}

\newcommand{\Mt}[1]{\mathbf{#1}}

\newcommand{\Norm}[1]{\left\Vert #1 \right\Vert}

\newcommand{\Set}[1]{\mathcal{\uppercase{#1}}}

\newcommand{\Transp}[1]{{#1}^{\mathrm{T}}}

\newcommand{\Vt}[1]{\mathbf{\lowercase{#1}}}

\newcommand{\mtN}{\Mt{N}}

\newcommand{\mtY}{\Mt{Y}}

\newcommand{\vtH}{\Vt{H}}

\newcommand{\vtW}{\Vt{W}}

\newcommand{\vtY}{\Vt{Y}}

\newcommand{\stA}{\Set{A}}

\newcommand{\stC}{\Set{C}}

\newcommand{\stK}{\Set{K}}

\newcommand{\stM}{\Set{M}}
\newcommand{\stN}{\Set{N}}

\newcommand{\stP}{\Set{P}}
\newcommand{\stQ}{\Set{Q}}

\newcommand{\stU}{\Set{U}}
\newcommand{\stV}{\Set{V}}

\newcommand{\bbC}{\mathbb{C}}

\newcommand{\bbE}{\mathbb{E}}

\newtheorem{lem}{Lemma}

\allowdisplaybreaks
\begin{document}
\title{Clustering Algorithms for Multi-CPU Cell-Free Massive MIMO Systems with Mixed Coherent and Non-Coherent Transmission}
\title{Clustering Algorithms for Multi-CPU Cell-Free Massive MIMO Systems with Mixed Transmission}
\title{Mixed Coherent and Non-Coherent Transmission for Multi-CPU Cell-Free Massive MIMO Systems}
\title{{Mixed Coherent and Non-Coherent Transmission for Multi-CPU Cell-Free Systems}}

\author{
	Roberto P. Antonioli$^{{\star},*}$,
	Iran M. Braga Jr.$^{\star}$,	
	G\'{a}bor Fodor$^{\dag\ddag}$,
	Yuri C. B. Silva$^\star$,
	Walter C. Freitas Jr.$^\star$\\
	\small Wireless Telecom Research Group (GTEL)$^\star$, Federal University of Cear\'a, Fortaleza, Brazil.\\
	\small Instituto Atlântico$^*$, Fortaleza, Brazil.
	\small Ericsson Research$^\dag$ and Royal Institute of Technology$^\ddag$, Stockholm, Sweden. \\
	\small \{antonioli, iran, yuri, walter\}@gtel.ufc.br, roberto\_antonioli@atlantico.com.br, gabor.fodor@ericsson.com, gaborf@kth.se
}

\maketitle

\begin{abstract}
Existing works on cell-free {systems} consider either coherent or non-coherent downlink
data transmission and a network deployment with a single \ac{CPU}. 
While it is known that coherent transmission outperforms non-coherent transmission when {assuming} unlimited fronthaul links,
{the former} requires a perfect timing synchronization, which is {practically not viable} over a large network.
Furthermore, relying on a single \ac{CPU} for geographically large cell-free networks is not scalable.
Thus, {to realize the expected gains of cell-free systems in practice,}
alternative transmission strategies for realistic multi-\ac{CPU} cell-free systems are required.
{Therefore, this paper}
proposes a novel downlink data transmission scheme that combines and generalizes the existing coherent and non-coherent transmissions.
The proposed transmission scheme, named mixed transmission, works based on the realistic assumption that only the \acp{AP} controlled by {a same} \ac{CPU} are synchronized, {and thus transmit in a coherent fashion}, while \acp{AP} from different \acp{CPU} {require no synchronism and transmit in a non-coherent manner}.
We also propose extensions of existing clustering algorithms for multi-\ac{CPU} cell-free systems with mixed transmission.
{Simulation results} 
show that the combination of the proposed clustering algorithms
with mixed transmission have the potential to perform close to the ideal coherent transmission.
\end{abstract}

\begin{IEEEkeywords} 
	Cell-free, clustering algorithms, mixed transmission. 
\end{IEEEkeywords}

\acresetall 

\section{Introduction}

The cell-free \ac{MIMO} architecture was conceived with the goal of eliminating cell boundaries and providing uniformly great service to every user in the system.
In the canonical form of the cell-free architecture, the network deployment follows a star topology, in which many distributed \acp{AP} have independent fronthaul connections to a single \ac{CPU}.
However, {it is unlikely} that practical deployments of cell-free systems in geographically large networks
will rely on a single \ac{CPU}, since the idea that the entire network would act as an infinitely large single cell,
comprised by a single \ac{CPU} and all the \acp{AP}, is not scalable.
Moreover, the common assumption that the data from all \acp{AP} would be processed in a coherent manner is not scalable~\cite{Interdonato2019} either.
In realistic large cell-free scenarios, the most likely deployment will comprise multiple \acp{CPU},
and each \ac{CPU} controls a disjoint set of \acp{AP}~\cite{Interdonato2019,Chen2021}.

Both in the canonical and multi-\ac{CPU} cell-free cases, achieving a sufficiently accurate relative timing
and phase synchronization, such that all \acp{AP} can jointly exploit coherent signal processing, is not a simple task.
Indeed, the communication theory underlying coherent transmissions for cell-free systems assumes that
a perfect timing synchronization exists, which is physically impossible over a large network,
even if the clocks from the \acp{AP} are synchronized~\cite{Chen2021}.
Besides that, joint coherent transmission requires that the data to be transmitted to a certain user
be simultaneously available at multiple \acp{AP}~\cite{Checko2015,Interdonato2019ubiquitous,Vu2020}.
Such strict requirements might be difficult to achieve even in centralized solutions~\cite{Alexandropoulos2016,Checko2015},
and {therefore it deserves} further investigation in the cell-free context~\cite{Demir2021,Chen2021}.

The majority of works on cell-free \ac{MIMO} have considered coherent downlink data transmission
in scenarios with unlimited fronthaul, such as in~\cite{Ngo2018,Interdonato2019,Bjornson2020,Papazafeiropoulos2021,Chakraborty2021,Saraiva2022}.
Assuming non-coherent downlink data transmission, a resource allocation for cell-free \ac{MIMO} is proposed in~\cite{Ammar2021}.
Some papers compared coherent and non-coherent downlink data transmissions in cell-free systems {under the assumption} of unlimited fronthauls~\cite{Ozdogan2019,Zheng2021}, {and concluded} that coherent transmission always outperforms 
non-coherent transmission in such scenarios.
In~\cite{Antonioli2022}, the authors compared coherent and non-coherent transmission considering the realistic assumption 
of limited {fronthauls}, and showed that non-coherent transmission outperforms coherent transmission 
in scenarios with low fronthaul capacities.

{While} the studies considering coherent downlink transmission in cell-free \ac{MIMO} systems 
proposed solutions that are transparent to the underlying cell-free topology, 
they assumed perfect network-wide synchronization, which is physically impossible~\cite{Chen2021}.
{Therefore, such an idealized coherent transmission can be only seen as an upper bound 
for the performance that can be achieved {in practice} by cell-free systems.}
{Thus, alternative transmission strategies that do not rely on network-wide synchronization 
are desired \cite{Varatharaajan:22}.}

In this work, we consider a {practically feasible} scenario with 
multiple \acp{CPU} controlling disjoint clusters of \acp{AP}, in which the multiple \acp{CPU} are not perfectly synchronized, 
such that the \acp{AP} from different \acp{CPU} are not capable of serving a certain \ac{UE} using coherent transmission. 
We assume {accurate} synchronization and simultaneous data availability {only within} the \acp{AP} from each \ac{CPU}.
In this context, we go beyond the existing works and propose a novel transmission scheme 
that combines and generalizes the existing coherent and non-coherent transmissions.
When considering multi-\ac{CPU} cell-free systems and the proposed mixed transmission, some existing algorithms need to be redesigned.
Along this line, we extend three existing clustering algorithms.
Computational simulations considering multi-\ac{CPU} cell-free systems with correlated fading, 
multi-antenna \acp{AP} and pilot contamination analyze the performance of the proposed clustering solutions 
and indicate that {the mixed transmission performs close to the coherent transmission.}

\textit{Notation:} Throughout the paper, matrices and vectors are presented by boldface upper and lower case letters, respectively.
{$\Transp{(\cdot)}$ and $\Herm{(\cdot)}$} denote the transpose and conjugate transpose operations, respectively.
$\{x_i\}_{\forall i}$ denotes the set of elements $x_i$ for the values of $i$ denoted by the subscript expression.
$\mathbf{I}$ is the identity matrix.
The cardinality of a discrete set $\mathcal{X}$ is denoted by $|\mathcal{X}|$.
Expected value of a random variable is denoted by $\mathbb{E}[\cdot]$.

\section{System Model}
\label{Clustering:sec:sys_model}

We consider the downlink of a cell-free system comprised of $M$ \acp{AP}, each equipped with $N$ antennas, and $K$ single-antenna \acp{UE}.
Let $\stM$ and $\stK$ be the sets of \acp{AP} and \acp{UE}, respectively, where $|\stM| = M$ and $|\stK| = K$.
The \acp{AP} are connected via fronthaul links to \acp{CPU}, where the set of \acp{CPU} is denoted as $\stQ$, with $|\stQ| = Q$.
The subset of \acp{AP} connected to \ac{CPU} $q$ is denoted by $\stV_q \subset \stM$.
We consider a user-centric cell-free system in which each \ac{UE}~$k$ is served only by a subset of nearby \acp{AP}, denoted herein as $\stA_{k} \subset \stM$, where $|\stA_{k}| = A_k$.
The subset of \acp{UE} served by \ac{AP} $m$ is denoted by $\stU_m \subset \stK$.

In the existing literature, it is commonly assumed that the set of \acp{AP} $\stA_{k}$ {serves} each \ac{UE}~$k$ using coherent transmission, in which the same data symbols are transmitted by all \acp{AP} in $\stA_{k}$ to \ac{UE}~$k$~\cite{Bjornson2020}.
Alternatively, some studies considered that the \acp{AP} in $\stA_{k}$ serve \ac{UE}~$k$ by means of non-coherent transmission, where the \acp{AP} in $\stA_{k}$ transmit different data symbols to \ac{UE}~$k$~\cite{Ozdogan2019,Zheng2021}.
When non-coherent transmission is considered, each single-antenna \ac{UE}~$k$ decodes the data symbols from different \acp{AP} using \ac{SIC}.

Herein, we go beyond the existing models from the literature and {develop} a mixed coherent-and-non-coherent transmission, 
denoted hereafter as mixed transmission, in which the set $\stA_{k}$ is {partitioned} into sets of \acp{AP} $\stA_{k}^{c}$.
It is assumed that ${\stA_{k}^{c} \cap \stA_{k}^{c'} = \emptyset, \; c \neq c'}$ and $\bigcup\limits_{\forall c} \stA_{k}^{c} = \stA_{k}$.
The set of coherent groups serving \ac{UE}~k is denoted as $\stC_k$, where $|\stC_k|=C_k \leq A_k$.
All \acp{AP} within set $\stA_{k}^{c}$ are synchronized and serve \ac{UE} $k$ in a coherent fashion, thus sending the same data symbols to \ac{UE}~$k$.
Each set $\stA_{k}^{c}$ is herein referred to as a coherent group.
While the \acp{AP} within $\stA_{k}^{c}$ send the same data symbols to \ac{UE}~$k$, the symbols transmitted by the \acp{AP} in $\stA_{k}^{c'}$ are different from the ones transmitted by the \acp{AP} in~$\stA_{k}^{c}$.
In this way, the \acp{AP} from set $\stA_{k}^{c}$ do not need to be synchronized with the \acp{AP} from set $\stA_{k}^{c'}$.
Therefore, 
{transmissions from different coherent groups occur in a non-coherent fashion}, 
such that the data symbols transmitted by different coherent groups need to be decoded by the \acp{UE} using \ac{SIC}.

The \acp{AP} and users operate {a \ac{TDD} protocol}, 
consisting of {an uplink} pilot phase used for channel estimation, 
{followed by a downlink data transmission phase.}
The channels estimated in the pilot phase are used by the \acp{AP} for precoding in the downlink data transmission phase, which is a valid assumption that relies on the channel reciprocity provided by \ac{TDD}~\cite{Interdonato2019ubiquitous,Demir2021}.

We assume that the channel vector $\vtH_{m,k} \in \bbC^{N\times1}$ between user $k$ and \ac{AP} $m$ follows a standard block fading model.
Therefore, $\vtH_{m,k}$ is constant in time-frequency blocks of $\tau_c$ symbols, where $\tau_c$ is the length of the coherence block.
Independent channel realizations are drawn in each coherent block following a correlated Rayleigh distribution~\cite{Bjornson2020}:
\begin{equation}
	\label{clustering:eq:channel_dl_rayleigh}
	\vtH_{m,k} \sim \stN_\bbC(\Mt{0},\Mt{R}_{m,k}),	
\end{equation}
where $\Mt{R}_{m,k} \in \bbC^{N\times N}$ is the spatial correlation matrix of the channel between user $k$ and \ac{AP} $m$, and $\beta_{m,k} \triangleq \frac{\text{tr}(\Mt{R}_{m,k})}{N}$ is the large scale fading (LSF) coefficient, which includes pathloss and shadowing.

\subsection{Pilot Phase and Channel Estimation}
\label{Clustering:sec:channel_estimation}

We assume that $\tau_p$ mutually orthogonal pilot sequences $\bm{\varphi}_1, \ldots, \bm{\varphi}_{\tau_p} \in \bbC^{\tau_p}$ are used for channel estimation.
Each user is assigned a pilot sequence during the initial access procedure, where such an assignment is done in a deterministic but arbitrary way and $\lVert\bm{\varphi}_t\rVert^{2} = \tau_p$.
In practical scenarios, it commonly happens that there are more users than available orthogonal pilot sequences in the system, i.e, $K > \tau_p$.
Thus, multiple users will share the same pilot sequence, giving rise to pilot contamination.
This situation is modeled herein using the set $\stP_k \subset \stK$, which represents the subset of users sharing the same pilot assigned to user $k$, including itself.
Also, let $t_k \in \{1,\ldots,\tau_p\}$ denote the index of the pilot sequence assigned to user $k$.

The received pilot signal matrix $\mtY_{m}^{\text{p}}\in\bbC^{N\times \tau_p}$ at the \ac{AP} $m$ is given by
\begin{align}
	\label{clustering:eq:pilot_signal_at_ap}
	\mtY_{m}^{\text{p}} = \sum_{k=1}^{K} \sqrt{p_{k}^{\text{p}}}\vtH_{m,k}\Herm{\bm{\varphi}}_{t_k} + \mtN_{m}^{\text{p}},
\end{align}
where $p_{k}^{\text{p}} \geq 0$ is the pilot power of user~$k$ and $\mtN_{m}^{\text{p}}\in\bbC^{N\times \tau_p}$
is the received noise at \ac{AP} $m$ with independent $\stN_\stC(0,\sigma^2)$ entries, where $\sigma^2$ is the noise power.

The channel estimation is conducted at \ac{AP} $m$ by projecting $\mtY_{m}^{\text{p}}$ onto the associated normalized pilot signal $\bm{\varphi}_{t_k}/\sqrt{\tau_p}$,
which results in $\check{\vtY}_{m,t_k} = \frac{1}{\sqrt{\tau_p}}\mtY_{m}^{\text{p}}\bm{\varphi}_{t_k} \in\bbC^{N}$.
Then, the \ac{MMSE} estimate of $\vtH_{m,k}$ is~\cite{Bjornson2020}
\begin{align}
	\label{clustering:eq:est_channel}
	\hat{\vtH}_{m,k} = \sqrt{p_k^\text{p}\tau_p}\Mt{R}_{m,k}\Psi_{m,t_k}^{-1}\check{\vtY}_{m,t_k},	
\end{align}
where
\begin{equation}
	\label{clustering:eq:correl_matrix_ycheck}
	\Psi_{m,t_k} = \bbE\left\lbrace \check{\vtY}_{m,t_k} \Herm{\check{\vtY}_{m,t_k}}\right\rbrace = \sum_{i \in \stP_k}\tau_p p_i^\text{p} \Mt{R}_{m,i} + \sigma^2 \Mt{I}_N
\end{equation}
is the correlation matrix of $\check{\vtY}_{m,t_k}$.
The estimated $\hat{\vtH}_{m,k}$ and the channel estimation error $\tilde{\vtH}_{m,k} = \vtH_{m,k} - \hat{\vtH}_{m,k}$ are independent vectors distributed as $\hat{\vtH}_{m,k} \sim \stN_\bbC(\Mt{0},p_k^\text{p}\tau_p\Mt{R}_{m,k}\Psi_{m,t_k}^{-1}\Mt{R}_{m,k})$ and $\tilde{\vtH}_{m,k} \sim \stN_\bbC(\Mt{0},\Mt{C}_{m,k})$, in which $\Mt{C}_{m,k} = \bbE\left\lbrace \tilde{\vtH}_{m,k} \Herm{\tilde{\vtH}_{m,k}}\right\rbrace = \Mt{R}_{m,k} - p_k^\text{p}\tau_p\Mt{R}_{m,k}\Psi_{m,t_k}^{-1}\Mt{R}_{m,k}$.
%

\subsection{Mixed Downlink Data Transmission and Closed-Form SE}
\label{Clustering:sec:mixed_signal_model}

Since this work focuses on the downlink transmission, we consider that $\tau_d = \tau_c - \tau_p$ symbols are reserved for downlink transmission.
The channels locally estimated at the \acp{AP} based on the \ac{MMSE} estimator given in~\eqref{clustering:eq:est_channel} are treated as the true channels, and are used for \ac{MR} beamforming during the downlink transmission, which is given by~\cite{Bjornson2020,Demir2021}
\begin{equation}
	\label{clustering:eq:mrt}
	\vtW_{m,k}^{\text{MR}} = \sqrt{\rho_{m,k}} \frac{\hat{\vtH}_{m,k}}{\sqrt{ \bbE\left\{\Norm{\hat{\vtH}_{m,k}}^2\right\} }},
\end{equation}
in which $\hat{\vtH}_{m,k}$ is the \ac{MMSE} channel estimate from~\eqref{clustering:eq:est_channel} and $\rho_{m,k} \geq 0$ is the data power from \ac{AP} $m$ to user $k$, with $\sum_{k \in \stU_m} \rho_{m,k} \leq \text{P}_m^\text{max}$, where $\text{P}_m^\text{max}$ is the power budget of \ac{AP}~$m$.

In the mixed downlink transmission, each \ac{AP} within the coherent group $\stA_{k}^c$ coherently transmits the same data symbol to the intended user $k$, while the \acp{AP} from different coherent groups $\stA_{k}^{c'}, c \neq c'$, transmit different symbols in a non-coherent fashion.
The signal received by user $k$ is given by:
\begin{equation}
	\label{clustering:eq:mixed_rec_signal}
	y_k \hspace{-.1cm}= \hspace{-.1cm}\sum_{c \in \stC_k} \sum_{m \in \stA_k^c} \Herm{\vtH}_{m,k} \vtW_{m,k} s_k^c + \mathop{\sum_{i \neq k}}_{i \in \stK} \sum_{c \in \stC_i} \sum_{m \in \stA_i^c} \Herm{\vtH}_{m,k} \vtW_{m,i} s_i^c + z_k,
\end{equation}
where $s_{k}^c\in\bbC$, with $\bbE\{|s_{k}^c|^2\} = 1$, is the transmitted data symbol transmitted by the \acp{AP} from $\stA_{k}^c$ to \ac{UE}~$k$, $z_k \sim \stN_\bbC(0,\sigma_k^{2})$ is the additive white Gaussian noise  at \ac{UE}~$k$, and $\vtW_{m,k}$ is the downlink beamforming vector from \ac{AP} $m$ to \ac{UE}~$k$.

To obtain a closed-form expression of the downlink \ac{SE} lower bound, we assume that the \acp{AP} use \ac{MR} precoding from~\eqref{clustering:eq:mrt} with the \ac{MMSE} channel estimation from~\eqref{clustering:eq:est_channel}, as shown next.

\begin{lem}
	\label{Lemma:RateMixed}
	Based on the received signal model in~\eqref{clustering:eq:mixed_rec_signal}, using \ac{MR} precoder in~\eqref{clustering:eq:mrt} and the \ac{MMSE} channel estimation in~\eqref{clustering:eq:est_channel}, a closed-form \ac{SE} achieved by the coherent group $c$ when transmitting to user $k$ is
	\begin{equation}
		\label{clustering:eq:closed_form_rate_CGc}
		r_k^c = \frac{\tau_d}{\tau_c} \log_2\left(1+\gamma_k^c\right),
	\end{equation}
	where
	\begin{equation}
		\label{clustering:eq:closed_form_sinr_CGc_mixed}
		\gamma_k^c = \frac{ \textup{D}_k^c }{\textup{E}_k  + \textup{F}_k - \sum_{b = 1}^{c}\textup{D}_k^b + \sigma_k^2},
	\end{equation}
	\begin{align}
		& \textup{D}_k^c = \left| \sum_{m \in \stA_k^c} \sqrt{\rho_{m,k} p_k^\text{p}\tau_p\tr\left(\Mt{R}_{m,k}\Psi_{m,t_k}^{-1}\Mt{R}_{m,k}\right)} \right|^2, \\
		& \textup{E}_k = \sum_{i \in \stK} \sum_{b \in \stC_i} \sum_{m \in \stA_i^b} \rho_{m,i} \frac{  \tr\left( \Mt{R}_{m,k} \Mt{G}_{i,i} \right) } { \tr\left(\Mt{G}_{i,i}\right) } ,\\
		& \textup{F}_k = \sum_{i \in \stK} \sum_{b \in \stC_i} \left| \sum_{m \in \stA_i^b}  \sqrt{\rho_{m,i} p_k^\text{p} \tau_p} \frac{   \tr\left(  \Mt{G}_{i,k} \right)  } { \sqrt{\tr\left(\Mt{G}_{i,i}\right)} } \right|^2 \left(\frac{\Herm{\bm{\varphi}}_{t_k} \bm{\varphi}_{t_i}}{\tau_p}\right), \\
		& \Mt{G}_{i,k} = \Mt{R}_{m,i} \Psi_{m,t_k}^{-1} \Mt{R}_{m,k}.
	\end{align}	
	Moreover, the closed-form total \ac{SE} achieved by user $k$ for the mixed transmission strategy is given by $r_k = \sum_{c \in \stC_k} r_k^c$.
\end{lem}

\begin{proof} {This proof can be conducted following~\cite[Section 6.2]{Demir2021}}.
\end{proof}

The rate expression derived in Lemma~\ref{Lemma:RateMixed} for the proposed mixed transmission is a generalization of the rate expressions of the existing coherent and non-coherent transmissions since the proposed mixed transmission particularizes to those two cases.
Specifically, observing \eqref{clustering:eq:closed_form_rate_CGc}, one can see that the mixed transmission reduces to the coherent case when there is only one large coherent group serving each user $k$ in a coherent fashion, i.e., $|\stC_k| = C_k = 1$ and $|\stA_k^1| = A_k$.
{This scenario where all \acp{AP} from the unique coherent group $\stA_k^1, \forall k \in \stK$ 
serve each user $k \in \stK$ in a coherent fashion would only be possible 
in an ideal situation, where all \acp{CPU} are perfectly synchronized.} 
{Note that} the mixed transmission rate expression in \eqref{clustering:eq:closed_form_rate_CGc} reduces to the non-coherent transmission if we consider that all coherent groups serving each user $k$ have only one \ac{AP}, i.e., $|\stC_k| = C_k = A_k$ and $|\stA_k^c| = 1, \forall c \in \stC_k$.

\section{Clustering Algorithms}
\label{Clustering:SEC:ClusteringAlgs}

This section proposes extensions of existing clustering algorithms -- which have been designed {for} single \ac{CPU} systems -- 
to 
{multi-\acp{CPU} cell-free systems using mixed transmissions.}
We remark that the clustering algorithms proposed herein {form clusters by adding one \ac{UE} at a time as proposed in ~\cite{Demir2021}}.
{However, our proposed schemes}
control the number of coherent groups {using mixed transmissions 
by setting a control parameter} $n_{\mathrm{CPU}}$, as detailed next.

\subsection{{LSF Larger Than Threshold}}

The first proposed algorithm selects the \acp{AP} from the $n_{\mathrm{CPU}}$ best \acp{CPU} with LSF larger than a predefined threshold.
The \acp{CPU} are ordered based on the following rule: 
(1) for each $k$, we select the \ac{AP} from each \ac{CPU} $q$ with the largest LSF coefficient $\beta_{m,k}, m \in \stV_q$; 
(2) then, the \acp{CPU} are ordered (in descending order) based on the LSF values of their best \ac{AP} with respect to user $k$; 
(3) after ordering the \acp{CPU}, we select the $n_{\mathrm{CPU}}$ best \acp{CPU} and consider only the \acp{AP} from those $n_{\mathrm{CPU}}$ best \acp{CPU} for the cluster formation, i.e., only $m \in \bigcup\limits_{q = 1}^{n_{\mathrm{CPU}}} \stV_q$.
Finally, after ordering the \acp{CPU}, the \acp{AP} meeting $\beta_{m,k} \geq \Delta$ are selected to serve user $k$, where $\Delta$ is a threshold of minimum required large scale fading quality.
This clustering algorithm is referred to as \textit{LSF algorithm} hereafter.
We highlight that setting $n_{\mathrm{CPU}}$ to the number of \acp{CPU} in the system corresponds to reducing the proposed solution to the legacy solution from~\cite{Bjornson2011}, in which each user could be served by all \acp{AP} in the system.

\subsection{{Fixed Number of APs}}

The second proposed algorithm selects a fixed number of \acp{AP} $n_{AP}$ from the $n_{\mathrm{CPU}}$ best \acp{CPU}. First, the \acp{CPU} are ordered based on the same rule of the previous solution and the values of $\beta_{m,k} \in \bigcup\limits_{q = 1}^{n_{\mathrm{CPU}}} \stV_q$ are sorted in descending order.
After ordering the \acp{CPU}, we select the $n_{\mathrm{CPU}}$ best \acp{CPU} and consider only the \acp{AP} from those $n_{\mathrm{CPU}}$ best \acp{CPU} for the cluster formation, i.e., only $m \in \bigcup\limits_{q = 1}^{n_{\mathrm{CPU}}} \stV_q$.
Then, the $n_{AP}$ \acp{AP} with largest LSF coefficients are selected to compose the cluster serving each user $k$.
This clustering algorithm is referred to as \textit{Fixed algorithm} hereafter.
This proposed solution reduces to the solution from~\cite{Buzzi2017} when setting $n_{\mathrm{CPU}}$ to the number of \acp{CPU} in the system.

\subsection{{Fraction of Total Received Power}}

The third proposed clustering algorithm selects the \acp{AP} from the $n_{\mathrm{CPU}}$ best \acp{CPU} to make sure that:
\begin{equation}
	\frac{\sum_{m \in \stA_k} \beta_{m,k}}{\sum_{m \in \bigcup\limits_{q = 1}^{n_{\mathrm{CPU}}} \stV_q} \beta_{m,k}} \geq \delta,
\end{equation}
where the \acp{CPU} are ordered based on the same rule of the previous solution and the values of $\beta_{m,k} \in \bigcup\limits_{q = 1}^{n_{\mathrm{CPU}}} \stV_q$ are sorted in descending order.
We remark that after ordering the \acp{CPU}, we select the $n_{\mathrm{CPU}}$ best \acp{CPU} and consider only the \acp{AP} from those $n_{\mathrm{CPU}}$ best \acp{CPU} for the cluster formation, i.e., only $m \in \bigcup\limits_{q = 1}^{n_{\mathrm{CPU}}} \stV_q$.
Thus, the main idea of this algorithm is to select the \acp{AP} from the $n_{\mathrm{CPU}}$ best \acp{CPU} to serve user $k$ that have the strongest channels and that contribute to more than $\delta\%$ of the total power received by user $k$.
This clustering algorithm is referred to as \textit{Power algorithm} hereafter.
This proposed multi-\ac{CPU} clustering solution reduces to the solution in~\cite{Ngo2018} if we set $n_{\mathrm{CPU}}$ to the number of \acp{CPU} in the system.

\section{Numerical Results and Discussions}
\label{Clustering:sec:ResultsDiscussion}


\subsection{Parameters and Setup}
\label{Clustering:SUBSEC:ParSetup}

We consider the downlink of a cell-free system in which \acp{AP} and users are uniformly distributed in a $1 \times 1$ km square, while the \acp{CPU} are positioned at fixed locations ($[250, 250]$~m, $[250, -250]$~m, $[-250, -250]$~m and $[-250, 250]$~m).
In the considered scenarios, each \ac{AP} is controlled by the closest \ac{CPU}.
By means of a wrap-around technique, we imitate an infinitely large network.
In terms of pilot assignment, we adopt a simple algorithm where each user randomly selects a pilot from a predefined set of orthogonal pilots~\cite{Chien2018,Ngo2018}.
The coherence blocks have $\tau_c = 200$ samples and, unless otherwise {stated}, the number of orthogonal pilots is $\tau_p=10$.
The pilot powers  are $p_{k}^{\text{p}} = 0.2$ W, $\forall k$ and the data powers are $\rho_{m,k} = 0.1$ W, $\forall m,k$, which are not optimized in this paper.

The large-scale coefficients are modeled by path loss and correlated shadowing (in dB): $\beta_{m,k} = \text{PL}_{m,k} + \sigma_{\text{sh}}\varkappa_{m,k}$, where $\sigma_{\text{sh}}\varkappa_{m,k}$ is the shadow fading
with standard deviation $\sigma_{\text{sh}}$ and $\varkappa_{m,k}\sim\stN(0,1)$,
while $\text{PL}_{m,k}$ represents the path loss given by the three slope model from~\cite{Ngo2018}.
For the shadow fading coefficients, we consider the {spatially correlated} model with two components.
The spatial correlation is computed via the Gaussian local scattering model with 15$^\circ$ angular standard deviation~\cite{Demir2021}.
Moreover, the noise power is given by $ \sigma^2 = B \kappa_{\text{B}} T_{0} \sigma_{\text{F}}$, where $B = $ 20 MHz, $\sigma_{\text{F}} = $ 9 dB, $\kappa_{\text{B}} = 1.381 \times 10^{-23}$ Joule per Kelvin and $T_{0} = 290$ Kelvin.

\subsection{Results and Discussions}
\label{Clustering:SUBSEC:ResDisc}

In the first simulations, we consider the following setup $\{M, K, N, A_k\} = \{100, 20, 2, 20\}$, where the cluster of \acp{AP} serving each user $k$ is formed by the $A_k$ \acp{AP} with the largest values of $\beta_{m,k}$, i.e., we adopt a largest-large-scale-based selection~\cite{Buzzi2017}.
This is a legacy solution that does not consider the multiple \acp{CPU} aspect in the system.

Fig.~\ref{FIG:CDF_NC_C_M_100_APs} compares the performance of the mixed, coherent and non-coherent transmissions.
The non-coherent transmission obtains the lowest values of total rate, while the coherent transmission achieves the highest values.
Meanwhile, the proposed mixed transmission achieves a total rate in between the existing transmission schemes, while being closer to the ideal coherent transmission.
The mixed transmission outperforms the non-coherent transmission because of the coherent groups from each \ac{CPU} that are formed to coherently transmit data to the users, which enhance the desired signal compared to the non-coherent transmission, in which each \ac{AP} independently transmits data symbols to the users (as if the coherent groups had size of 1).
However, even though the formation of the coherent groups {enhances} the system performance 
when using the mixed transmission compared to the non-coherent transmission, 
{this} performance increase is not enough to reach the performance of the ideal coherent transmission.
This occurs because the mixed transmission has some degree of non-coherent transmission 
from the different coherent groups, which is known to have {inferior performance  to coherent transmission}.
Meanwhile, {for coherent transmission}, 
all \acp{AP}  must be perfectly synchronized to transmit data to the users.

\begin{figure}[!h]
	\centering
	\includegraphics[width=.8\columnwidth]{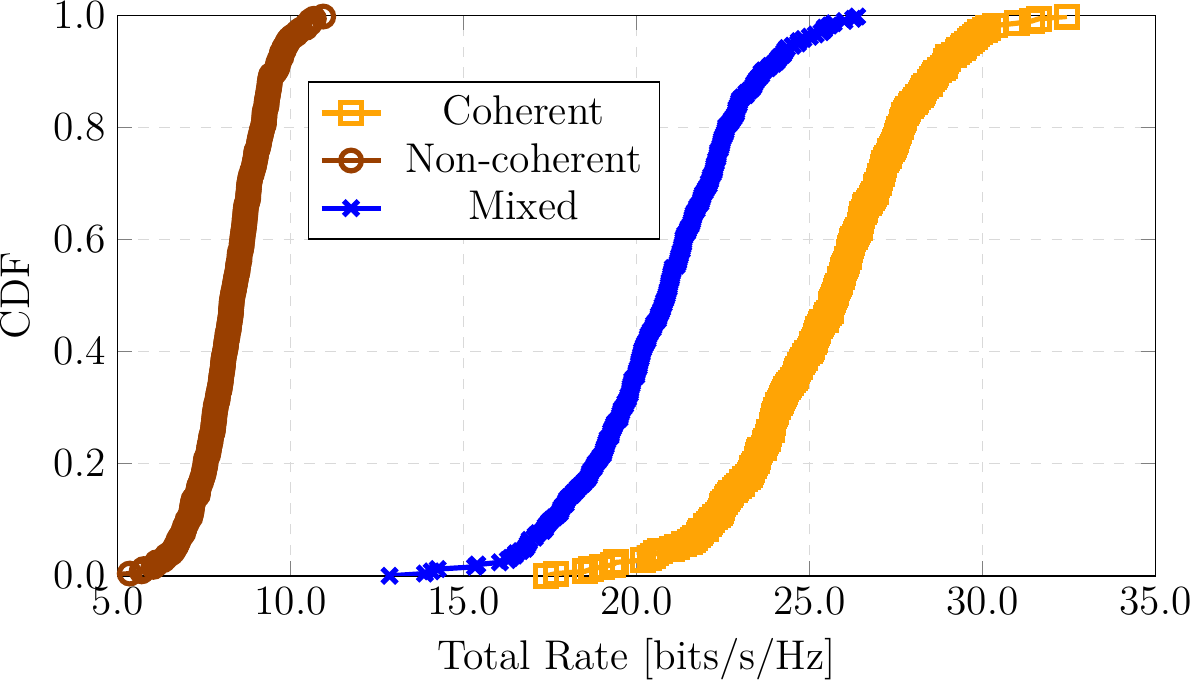}%
	\caption{CDF of total rate comparing the coherent, non-coherent and mixed transmissions.}
	\label{FIG:CDF_NC_C_M_100_APs}
\end{figure}

We now analyze the impact of $A_k$ while still using the largest-large-scale-based selection from~\cite{Buzzi2017} as clustering algorithm.
The parameters used in this {analysis} are $\{M, K, N\} = \{100, 20, 2\}$.

Fig.~\ref{FIG:NC_vs_C_vs_M_100APs_Varying_Ak} shows the total system rate obtained by coherent, non-coherent and mixed transmissions when the value of $A_k$ is varied.
When $A_k=1$, all three {transmission} schemes present the same performance because the rate expressions for all three schemes reduce to the same expression.
As $A_k$ increases, the rates achieved by the non-coherent transmission steadily decrease.
This occurs because of the increasing interference in the system, which increases faster than the desired signal strength due to the \ac{SIC} decoding.
On the other hand, both mixed and coherent transmissions increase the total rate up to a certain value of $A_k$, from which point the total rates achieved by these two schemes start decreasing.
The reason behind the increase in the total rate achieved by the mixed and coherent transmissions for the smaller values of $A_k$ is because they are able to increase the derived signal strength more than the interference increases.
This thinking is valid up to a certain value of $A_k$, from which the increase in the desired signal strength does not compensate the increase in the interference and pilot contamination, thus decreasing the total rate.

\begin{figure}[!h]
	\centering
	\includegraphics[width=.8\columnwidth]{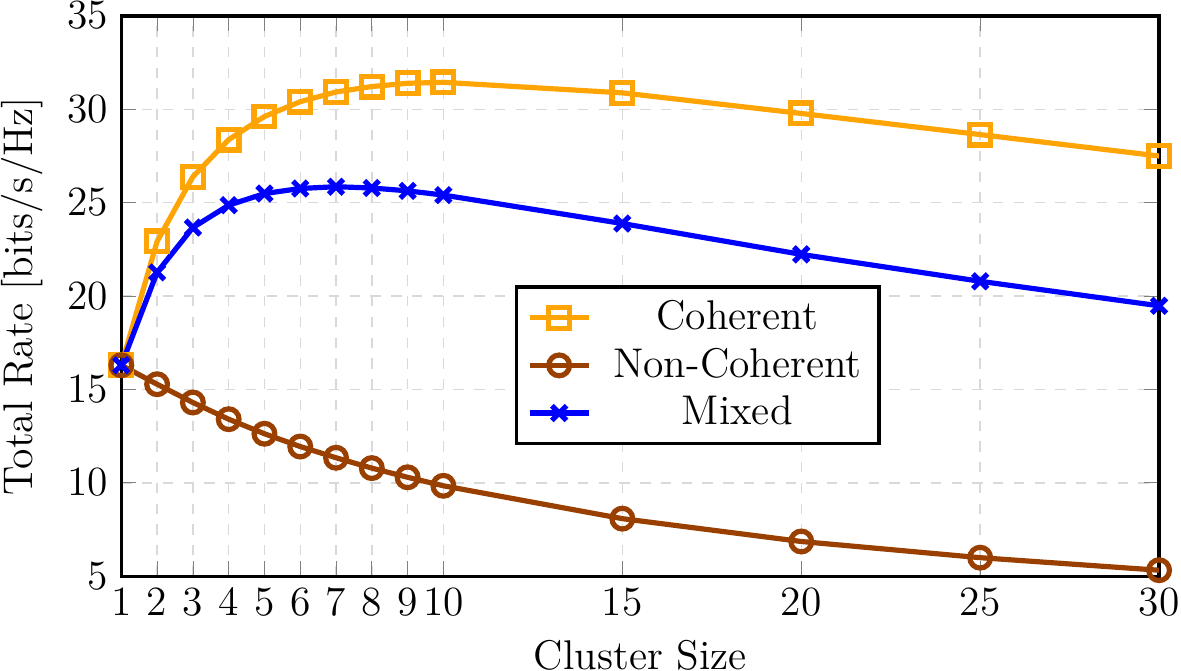}
	\caption{Total rate comparing coherent, non-coherent and mixed transmissions when varying~$A_k$.}
	\label{FIG:NC_vs_C_vs_M_100APs_Varying_Ak}
\end{figure}

The next figures analyze the performance of the coherent, non-coherent and mixed transmissions when adopting the clustering algorithms proposed in Section~\ref{Clustering:SEC:ClusteringAlgs}.
The parameters used for the clustering algorithms are: $n_{\mathrm{CPU}} = \left\{1, 2, 4\right\}$, $\Delta = \left\{23.5, 64.36, 266.06\right\}$, $\delta = \left\{0.85, 0.90, 0.95\right\}$ and $n_{\mathrm{AP}} = \left\{5, 10, 15\right\}$.
The system parameters used in this subsection are $\{K, N\} = \{20, 2\}$, while the number of \acp{AP} is varied.

Fig.~\ref{FIG:Sum_Rate_VaryAP_AlgEight} shows the total system rate obtained when adopting the Power algorithm. For the mixed transmission, we selected $\delta = 0.95$, which yields the best performance, while varying the parameter $n_{\mathrm{CPU}}$.
For the coherent transmission, the best performance occurs when $\delta = 0.95$ and $n_{\mathrm{CPU}} = 2$, which means that the \acp{AP} from the 2 best \acp{CPU} will transmit in a coherent fashion to the users.
Meanwhile, for the non-coherent transmission, the best performance occurs when $\delta = 0.85$ and $n_{\mathrm{CPU}} = 1$.
It is worth noting that none of the transmission strategies presented their best performances with $n_{\mathrm{CPU}} = 4$, i.e., when using the legacy solution, which shows the importance of multi-\ac{CPU}-aware clustering solutions.
Note that the non-coherent transmission presents its best performance when $\delta = 0.85$, i.e., when the cluster of \acp{AP} assumes the smallest size considering the analyzed parameters.
Meanwhile, both coherent and mixed transmissions present their best performance when creating a larger cluster of \acp{AP} in a controlled manner by setting $\delta = 0.95$.
Moreover, the mixed transmission strategy presented the best performance when the clustering algorithm creates only one large coherent group (i.e., when $n_{\mathrm{CPU}} = 1$), which indicates that the total system rate is not enhanced when adopting multiple non-coherent transmissions from different coherent groups.
Also, when $\delta = 0.95$ and $n_{\mathrm{CPU}} = 1$, the proposed practical mixed transmission strategy achieved a total rate very close to the ideal full coherent transmission.

\begin{figure}[!h]
	\centering
	\includegraphics[width=.8\columnwidth]{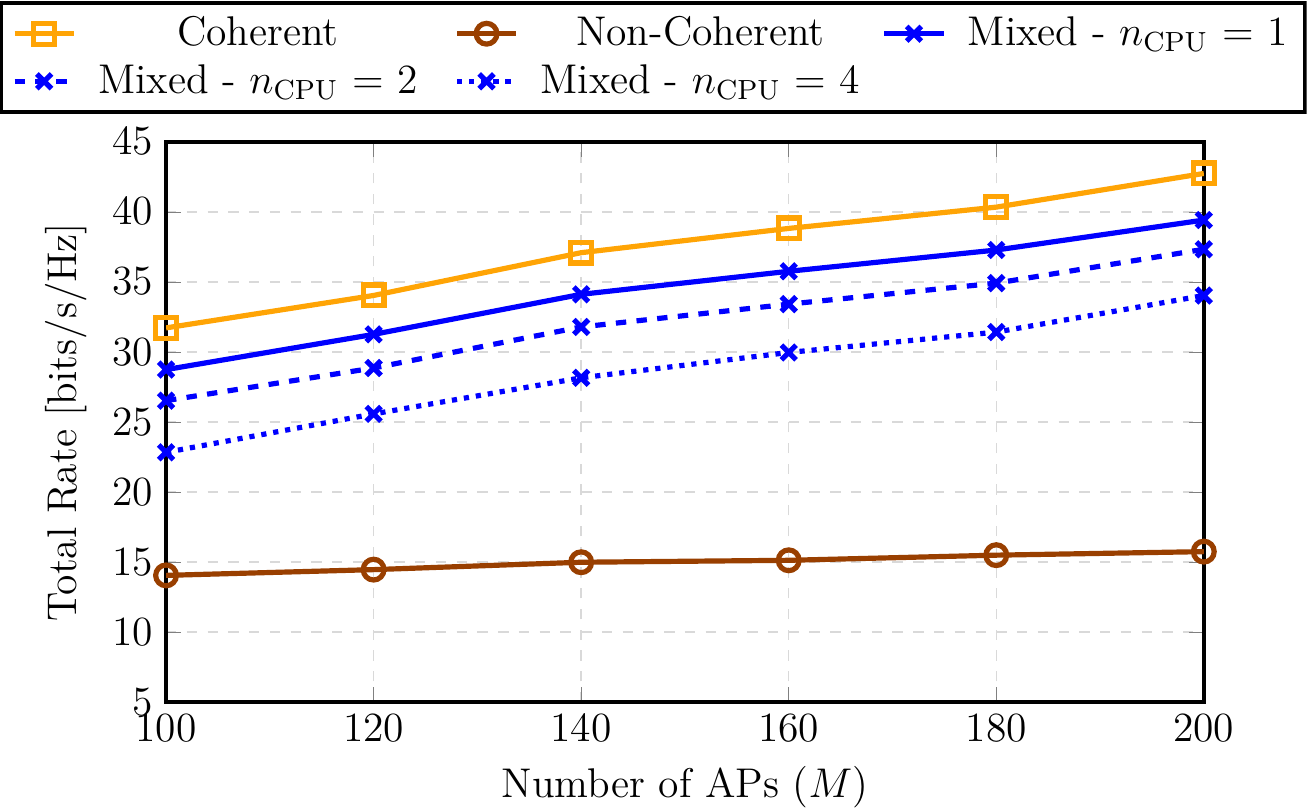}
	\caption{Impact of the number of different coherent groups on the performance of the mixed transmission with the \textit{Power algorithm}.}
	\label{FIG:Sum_Rate_VaryAP_AlgEight}
\end{figure}

Fig.~\ref{FIG:Sum_Rate_VaryAP_AlgNine} shows the total rate obtained when adopting the Fixed algorithm. For the mixed transmission, we selected  $n_{\mathrm{AP}} = 15$, which yields the best performance, while varying the parameter $n_{\mathrm{CPU}}$.
For the coherent transmission, the best performance occurs when $n_{\mathrm{AP}} = 15$ and $n_{\mathrm{CPU}} = 2$. Finally, for the non-coherent transmission, the best performance occurs when $n_{\mathrm{AP}} = 5$ and $n_{\mathrm{CPU}} = 1$.
Note again {that} the analyzed transmission schemes {do not reach} their best performance 
when $n_{\mathrm{CPU}} = 4$.
As observed previously, the non-coherent transmission {achieves} its highest performance when using the smallest cluster of \acp{AP}.
Meanwhile, the same conclusions for the mixed transmission discussed for Fig.~\ref{FIG:Sum_Rate_VaryAP_AlgEight} apply for Fig.~\ref{FIG:Sum_Rate_VaryAP_AlgNine}.

\begin{figure}[!h]
	\centering
	\includegraphics[width=.8\columnwidth]{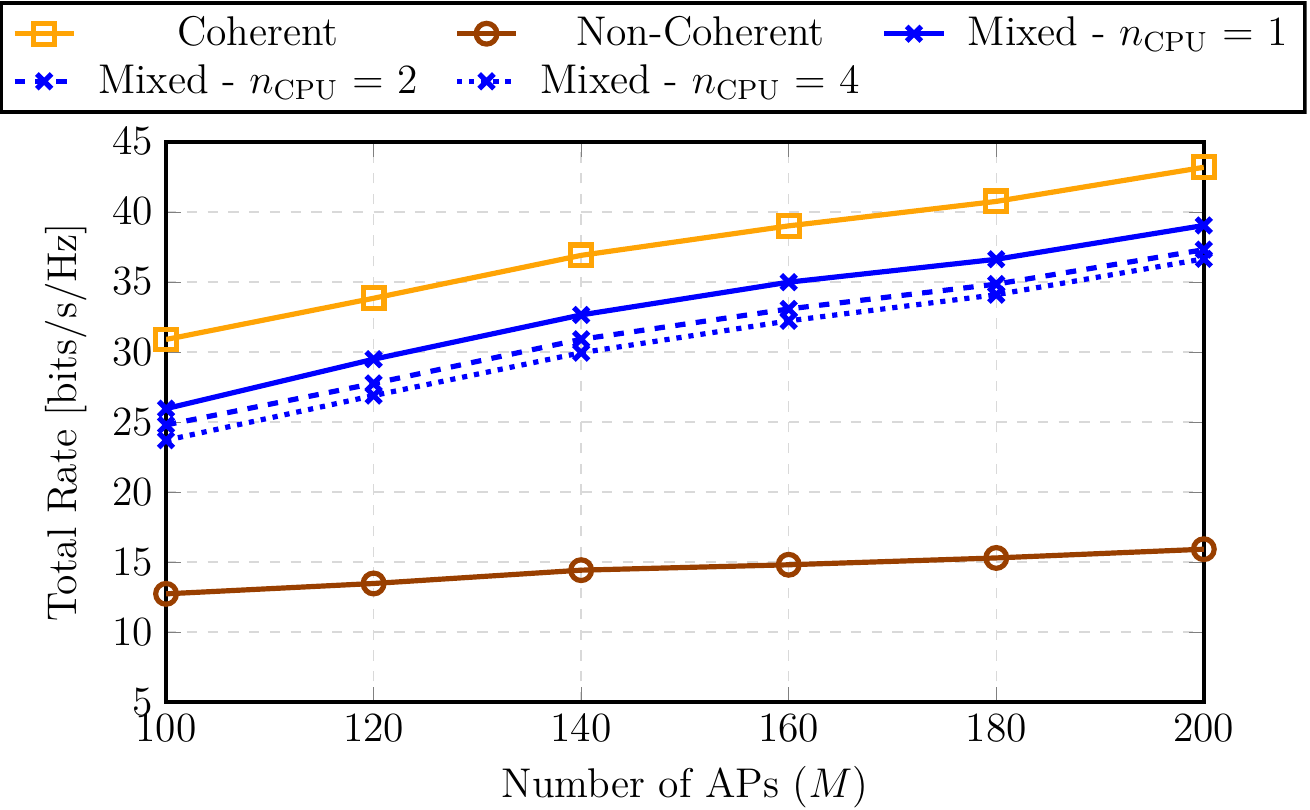}
	\caption{Impact of the number of different coherent groups on the performance of the mixed transmission with the \textit{Fixed algorithm}.}
	\label{FIG:Sum_Rate_VaryAP_AlgNine}
\end{figure}

Fig.~\ref{FIG:Sum_Rate_VaryAP_AlgTen} shows the total rate obtained when adopting the LSF algorithm. For the mixed transmission, we selected $\Delta = 23.5$, which yields the best performance, while varying the parameter $n_{\mathrm{CPU}}$.
For the coherent transmission, the best performance occurs when $\Delta = 23.5$ and $n_{\mathrm{CPU}} = 4$.
Meanwhile, for the non-coherent transmission, the best performance occurs when $\Delta = 266.06$ and $n_{\mathrm{CPU}} = 1$.
In particular, for the LSF algorithm, the best performance presented by the coherent transmission was achieved when $n_{\mathrm{CPU}} = 4$, i.e., when using the legacy solution from~\cite{Bjornson2011}.
The other conclusions from Fig.~\ref{FIG:Sum_Rate_VaryAP_AlgEight} and Fig.~\ref{FIG:Sum_Rate_VaryAP_AlgNine} apply to Fig.~\ref{FIG:Sum_Rate_VaryAP_AlgTen}.

\begin{figure}[!h]
	\centering
	\includegraphics[width=.8\columnwidth]{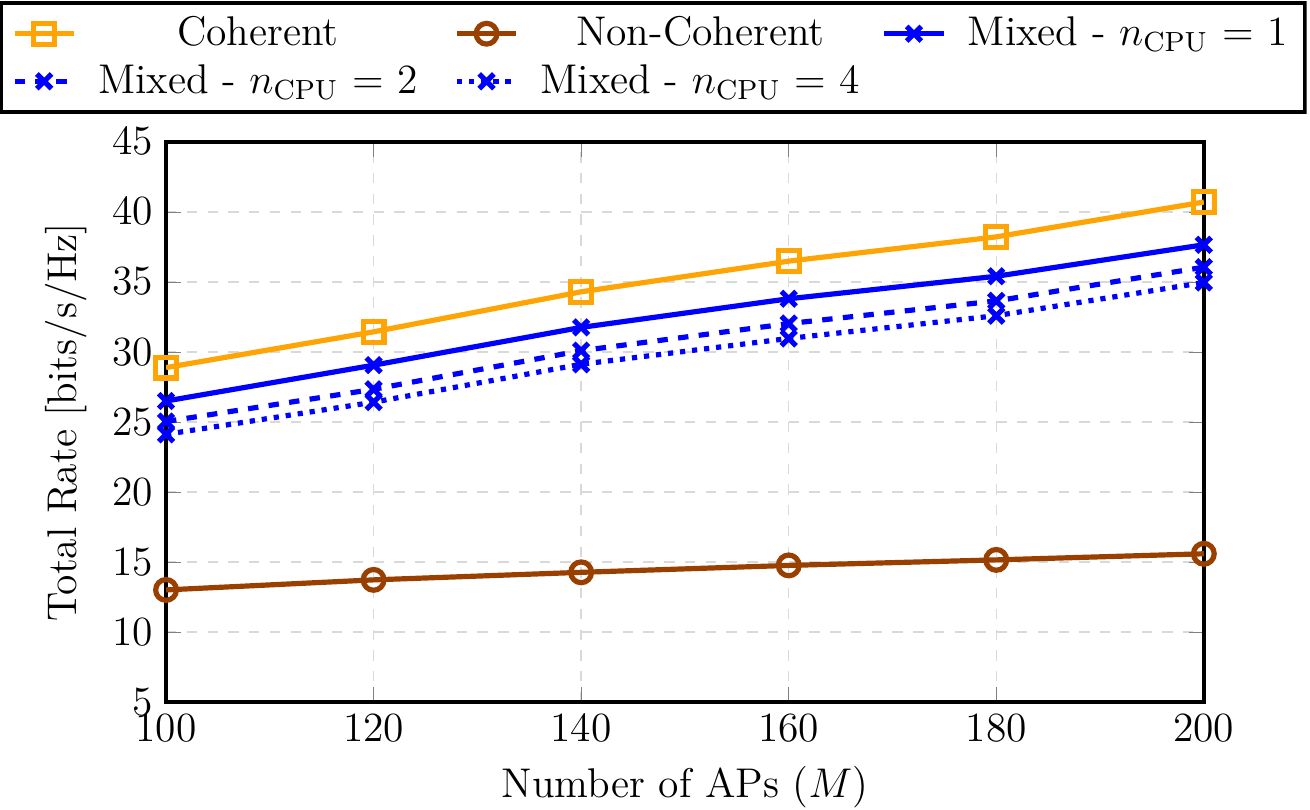}
	\caption{Impact of the number of different coherent groups on the performance of the mixed transmission with the \textit{LSF algorithm}.}
	\label{FIG:Sum_Rate_VaryAP_AlgTen}
\end{figure}

Fig.~\ref{FIG:Sum_Rate_VaryAP_AllAlgs} compares the best total rates achieved by the proposed multi-CPU-aware clustering solutions and transmission schemes.
The performance obtained by the Power and Fixed algorithms are very close for all transmission schemes, with a small advantage to the Power algorithm.
The small gain achieved by the Power algorithm comes from its smart \ac{AP} selection based on a fraction of the received power by each user.
Furthermore, the mixed transmission strategy using the Power algorithm presents almost the same performance of the coherent transmission using the LSF algorithm, showing the potential of the mixed transmission scheme and the importance of properly forming the cluster of \acp{AP} to serve the users.

\begin{figure}[!h]
	\centering
	\includegraphics[width=.8\columnwidth]{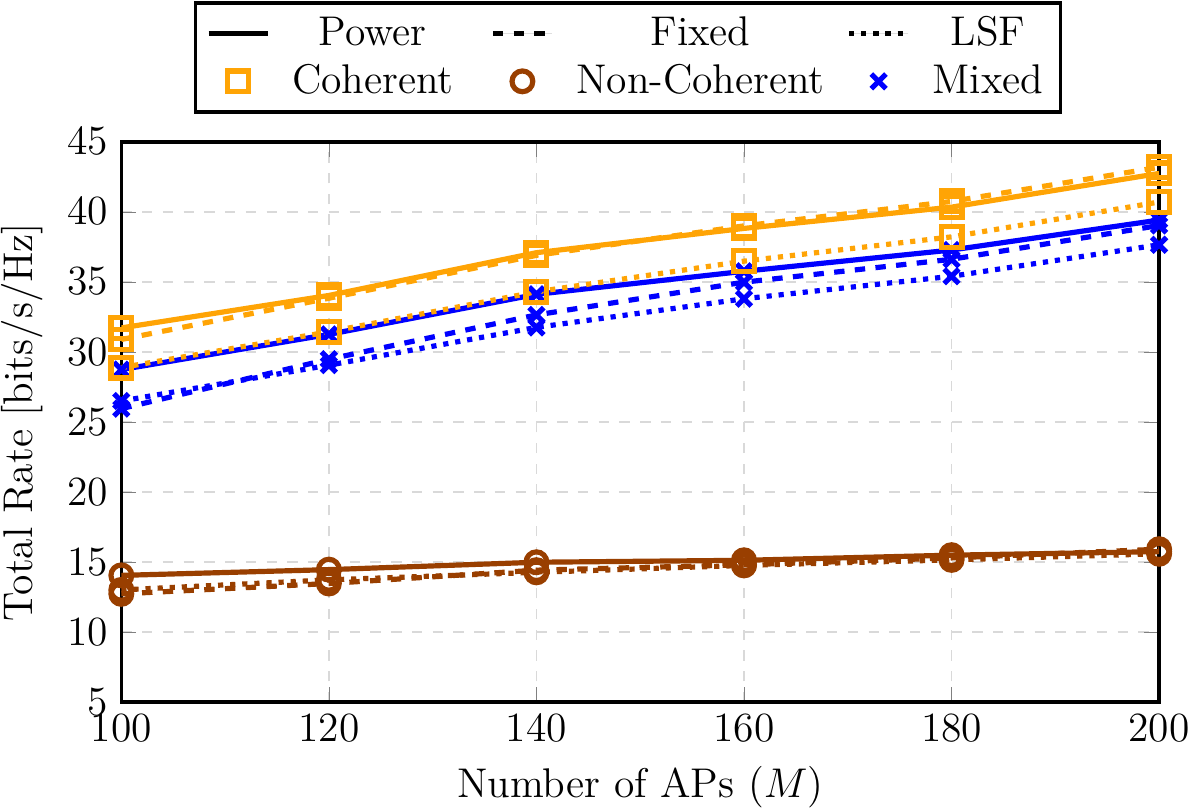}
	\caption{Comparison of the sum-rate achieved by all clustering algorithms considering their best performances for each transmission scheme.}
	\label{FIG:Sum_Rate_VaryAP_AllAlgs}
\end{figure}

\section{Conclusions}
\label{Clustering:sec:Conclusion}

In this work, we studied multi-\ac{CPU} cell-free systems with mixed transmission and the impact of the clustering algorithms in these systems.
We concluded that legacy clustering algorithms fail to exploit the benefits of the mixed transmission and that the design of multi-\ac{CPU}-aware clustering algorithms successfully improve the system performance.
More specifically, when using the proposed mixed transmission, higher total rates are obtained when forcing the users to be connected to only one \ac{CPU} and exploiting the coherent transmission potential from that \ac{CPU}.
Also, it was shown that the proposed practical mixed transmission achieves total rates very close to the ideal full coherent transmission, validating its practical {viability}.

\section{Acknowledgment}

This work was supported in part by Ericsson Research,
Technical Cooperation Contract UFC.48, in part by CNPq,
in part by FUNCAP, in part by CAPES/PRINT
Grant 88887.311965/2018-00, in part by CAPES - Finance
Code 001 and in part by the Celtic project 6G for
Connected Sky, Project ID: C2021/1-9.

\bibliographystyle{IEEEtran}
\bibliography{IEEEabrv,cfrefs}

\end{document}